\documentclass[hyphens]{article}
\usepackage{newtxtext}
\usepackage{helvet}
\usepackage{courier}
\usepackage{graphicx} 
\usepackage{bm}
\usepackage[linesnumbered,lined,ruled,noend,vlined]{algorithm2e}

\usepackage{latexsym}
\usepackage{amssymb,mathtools}
\usepackage{amsfonts}
\usepackage{amsthm}
\usepackage[hidelinks]{hyperref}
\usepackage[nameinlink,capitalize]{cleveref}
\usepackage[margin=1in]{geometry}

\title{Frameworks to Design Approximation Algorithms for Finding Diverse Solutions in Combinatorial Problems\thanks{
A preliminary version of this paper appeared in the Proceedings of the 37th AAAI Conference on Artificial Intelligence, 2023~\cite{10.1609/aaai.v37i4.25511}.
Compared with the conference version, this version includes full proofs, expanded discussions, and additional results presented in Section~\ref{sec:random}.
}}

\author{
Tesshu Hanaka\thanks{Kyushu University, Fukuoka, Japan, \email{hanaka@inf.kyushu-u.ac.jp}} \and
Masashi Kiyomi\thanks{Seikei University, Tokyo, Japan, \email{kiyomi@st.seikei.ac.jp}} \and
Yasuaki Kobayashi\thanks{Hokkaido University, Sapporo, Japan, \email{koba@ist.hokudai.ac.jp}} \and
Yusuke Kobayashi\thanks{Kyoto University, Kyoto, Japan, \email{yusuke@kurims.kyoto-u.ac.jp}} \and
Soh Kumabe\thanks{CyberAgent, Tokyo, Japan, \email{kumabe\_soh@cyberagent.co.jp}} \and
Kazuhiro Kurita\thanks{Okayama University, Okayama, Japan, \email{k-kurita@okayama-u.ac.jp}} \and
Yota Otachi\thanks{Nagoya University, Nagoya, Japan, \email{otachi@nagoya-u.jp}}
}

\newcommand{\email}[1]{\texttt{#1}}

\usepackage{mathrsfs}

\newcommand{\set}[1]{\{#1\}}
\newcommand{\inset}[2]{\{#1 : #2\}}
\newcommand{\size}[1]{\left|#1\right|}
\newcommand{\order}[1]{O(#1)}

\newtheorem{lemma}{Lemma}
\newtheorem{theorem}[lemma]{Theorem}
\newtheorem{corollary}[lemma]{Corollary}

\newtheorem{definition}[lemma]{Definition}
\newtheorem{observation}{Observation}
\newcommand{\E}{\mathbf{E}}

\newcommand{\dsum}[1]{d_{\rm sum}(#1)}

\newcommand{\MSDC}{\textsc{Max-Sum Diverse Solutions}}
\newcommand{\NP}{\ensuremath{\mathrm{NP}}}

\makeatletter
\@ifpackageloaded{algorithm2e}{
\SetKwProg{Fn}{Function}{}{}
\SetKwProg{Procedure}{Procedure}{}{}
\SetKwProg{Subprocedure}{Subprocedure}{}{}
\SetKw{Continue}{continue}
\SetKwComment{tcc}{//}{}
\SetKwFunction{Local}{LocalSearch}%
\SetKwFunction{LocalDup}{LocalSearch-Dup}%
\SetKwFunction{Top}{Top-r}%
\SetKwFunction{Greedy}{Greedy}%
\SetKwFunction{FindMaxSum}{FindMaxSum}%
\SetKwFunction{FindMaxMin}{FindMaxMin}%
\SetKw{Continue}{continue} 
\SetKwInOut{AlgInput}{Input}
\SetKwInOut{AlgOutput}{Output}
\SetKwInOut{AlgPrecondition}{Pre-conditions}
\SetKwInOut{AlgInvariant}{Invariants}
}{}
\makeatother

\crefname{lemma}{Lemma}{Lemmas}
\crefname{theorem}{Theorem}{Theorems}
\crefname{corollary}{Corollary}{Corollaries}
\crefname{figure}{Figure}{Figures}
\crefname{algorithm}{Algorithm}{Algorithms}

\date{}

\begin{document}

\maketitle

\begin{abstract}
    Finding a \emph{single} best solution is the most common objective in combinatorial optimization problems.
    However, such a single solution may not be applicable to real-world problems as objective functions and constraints are only ``approximately'' formulated for original real-world problems.
    To solve this issue, finding \emph{multiple} solutions is a natural direction, and diversity of solutions is an important concept in this context.
    Unfortunately, finding diverse solutions is much harder than finding a single solution.
    To cope with the difficulty, we investigate the approximability of finding diverse solutions.
    As a main result, we propose a framework to design approximation algorithms for finding diverse solutions, which yields several outcomes including constant-factor approximation algorithms for finding diverse matchings in graphs and diverse common bases in two matroids and PTASes for finding diverse minimum cuts and interval schedulings.
\end{abstract}

\section{Introduction}

One way to solve a real-world problem is to formulate the problem as a mathematical optimization problem and find a solution with an optimization algorithm.
However, it is not always easy to formulate an appropriate optimization problem as real-world problems often include intricate constraints and implicit preferences, which are usually simplified in order to solve optimization problems.
Hence, an optimal solution obtained in this way is not guaranteed to be a ``good solution'' to the original real-world problem.
To cope with this underlying inconsistency, the following two-stage approach would be promising: algorithms find multiple solutions and then users find what they like from these solutions.
One may think that top-$k$ enumeration algorithms (see~\cite{Eppstein:EA:2008} for a survey) can be used for this purpose.
However, this is not always the case since top-$k$ enumeration algorithms may output solutions similar to one another. (See~\cite{Wang:KDD:2013,Yuan:VLDBJ:2016,Hao:FGCS:2020}, for example).
Such a set of solutions is not useful as a ``catalog'' of solutions provided to users.

As a way to resolve this issue, algorithms are expected to find ``diverse'' solutions, and then finding ``diverse'' solutions has received considerable attention in several fields such as Artificial Intelligence~\cite{Ingmar:AAAI:2020,Nadel:SAT:2011}, Data Mining~\cite{Wang:KDD:2013,Yuan:VLDBJ:2016}, and Operations Research~\cite{Danna:ORL:2009,Petit:INFORMS:2019}.
The problem of finding diverse solutions can be modeled as a multi-objective optimization problem, which optimizes some diversity measure and the quality of solutions simultaneously.
To solve this multi-objective optimization problem, there are several approaches, such as mathematical programming~\cite{DannaFGW:IPCO:2007,Danna:ORL:2009,Petit:INFORMS:2019}, constraint programming~\cite{Hebrard:AAAI:2005,Petit:IJCAI:2015}, heuristics~\cite{DannaFGW:IPCO:2007,Drosou:SIGMODREC:2010,HentenryckCG:CP:2009,VieiraRBHSTT:ICDE:2011}, and so forth. See Table~1 in \cite{Petit:INFORMS:2019}, which summarizes various approaches to find diverse solutions in the literature.
These approaches are not only practical but also versatile, enabling us to formulate various combinatorial problems in these approaches.
However, theoretical guarantees on the running time of algorithms and/or the quality of solutions would be difficult to obtain as some (general purpose) mathematical/constraint solvers or heuristic objectives/algorithms are key components of these algorithms.

Recently, theoretical aspects of the problem of finding diverse solutions in combinatorial problems are investigated.
This research direction would be made by Fellows and Rosamond who proposed the \emph{diverse X paradigm} in Dagstuhl Seminar 18421~\cite{Fernau:DR:2019}.
In this paradigm, ``X'' is a placeholder that represents solutions we are looking for, and they asked for theoretical investigations of finding diverse solutions.
Since the problem of finding diverse solutions is much harder than that of finding a single solution for some ``X'', it would be reasonable to consider the problem from the perspective of fixed-parameter tractability\footnote{Roughly speaking, the goal is to develop algorithms that run in time $f(k){\rm poly}(n)$, where $n$ is the input size and $k$ is a parameter defined on a specific problem.}.
From this proposition, several fixed-parameter tractable (FPT) algorithms are developed so far. 
Baste et al.\ gave algorithms for finding diverse solutions related to hitting sets~\cite{Baste:ALGO:2019} and those on bounded-treewidth graphs~\cite{Baste:AI:2022}.
Hanaka et al.~\cite{Hanaka:AAAI:2021} proposed a framework to obtain FPT algorithms for finding diverse solutions in various combinatorial problems.
Fomin et al.~\cite{Fomin:ISAAC:2020,Fomin:STACS:2021} investigated the fixed-parameter tractability of finding diverse solutions related to matchings and matroids.
In these work, the number of solutions to be found is considered as a small parameter, which can be a potential drawback in practice.
As we discussed, a set of diverse solutions would be displayed as a ``catalog'' of solutions and hence moderate number of solutions are essential to users to make their own decisions based on the displayed solutions.

For this reason, we aim to develop theoretically efficient algorithms for finding a moderate number of diverse solutions rather than a small number of diverse solutions. 
As we mentioned, the problem of finding diverse solutions is harder than that of finding a single solution.
We first observe that diversity measures have a significant impact on the computational complexity of the diverse version of combinatorial problems:
The problem of computing $k$ bases of a matroid maximizing the minimum (weighted) Hamming distance (\textsc{Max-Min Hamming Distance}) is \NP-hard~\cite{Fomin:STACS:2021}, while the problem maximizing the sum of (weighted) Hamming distance (\textsc{\textsc{Max-Sum Hamming Distance}}) is solvable in polynomial time~\cite{Hanaka:AAAI:2021}.
Hanaka et al. \cite{Hanaka:AAAI:2022} enhanced this observation by showing that the diverse versions of several classical combinatorial problems, such as bipartite matchings, arborescences, shortest paths, are polynomial-time solvable under \textsc{Max-Sum Hamming Distance}, while no such results for \textsc{Max-Min Hamming Distance} are known.
These circumstances indicate that \textsc{Max-Sum Hamming Distance} is theoretically easier than \textsc{Max-Min Hamming Distance}.
However, there are still computationally hard problems under \textsc{Max-Sum Hamming Distance}: For example, the problem of computing a maximum matching in a graph is known to be solvable in polynomial time, whereas that of computing two maximum matchings $M_1$ and $M_2$ maximizing $|M_1 \mathbin{\triangle} M_2|$ is known to be \NP-hard~\cite{Holyer81a} (see Section~\ref{sec:appl} for other examples).
Thus, we tackle this intractability by developing \emph{polynomial-time approximation algorithms} for the diverse version of various combinatorial problems.
To this end, we employ \textsc{Max-Sum Hamming Distance} as our diversity measure (see \Cref{sec:prel} for its definition), which might be somewhat theoretically tractable, but there are still several obstacles to be overcome.
We note that this diversity measure is frequently used in both experimental and theoretical settings~\cite{Baste:AI:2022,Danna:ORL:2009,Hanaka:AAAI:2022,HentenryckCG:CP:2009,Petit:IJCAI:2015} 

Our main result is a framework for designing efficient approximation algorithms with constant approximation factors for finding diverse solutions in combinatorial problems.
Roughly speaking, our approximation framework says that if we can \emph{enumerate} top-$k$ \emph{weighted} solutions in polynomial time, then we can obtain in polynomial time \emph{unweighted} solutions maximizing our diversity measure with constant approximation factors. 
Moreover, suppose that we can exactly maximize our diversity of solutions in polynomial time when the number of solutions we are looking for is bounded by a constant.
Then, our framework yields a polynomial-time approximation scheme (PTAS), meaning that factor-$(1 - \varepsilon)$ approximation in polynomial time for every constant $\varepsilon > 0$.
By applying our framework, we obtain efficient constant-factor approximation algorithms for finding diverse matchings in a graph and common bases of two matroids, while PTASes for finding diverse minimum cuts and interval schedulings.
Let us note that these diversity maximization problems are unlikely to be solvable in polynomial time, which will be discussed later.

The approximation factor of our framework comes from previous work on the dispersion problem (see \Cref{sec:prel} for its definition).
A similar framework was independently proposed by Gao et al.~\cite{Gao:LATIN:2022}.
In both frameworks, the subproblem of finding a ``furthest solution'' is a key ingredient.
They reduced this subproblem to the budget-constrained optimization problem and solve it by bi-approximation algorithms.
This makes their framework more flexible, allowing to find diverse approximate weighted solutions, while our framework only focus on unweighted solutions.
As opposed to this weight restriction, the approximation factor of our framework is much better than theirs: their approximation factor is $1/2$, while ours are $\max(1-1/k, 1/2)$ or even $1 - \varepsilon$ for any constant $\varepsilon>0$, where $k$ is the number of solutions we are looking for.

The rest of this paper is organized as follows.
The next section gives some notation and terminology used in this paper.
In particular, we give an overview of the result of \cite{Cevallos:MOR:2019}, which is a key to our approximation algorithms.
In \Cref{sec:framework}, we describe our framework to find diverse solutions in combinatorial problems.
Then, in \Cref{sec:appl}, we discuss some applications of our framework, including the diverse versions of the maximum matching problem, the matroid intersection problem, the minimum cut problem, and the interval scheduling problem.
In \Cref{sec:random}, we improve the approximation factor of our framework to $(1 - \varepsilon)(1 - 1/k)$ in a setting that allows duplications.
Finally, we conclude our paper with some further directions in \Cref{sec:concl}.

\section{Preliminaries}\label{sec:prel}
We denote the set of real numbers, the set of non-negative real numbers, and the set of positive real numbers as $\mathbb R$, $\mathbb R_{\ge 0}$, and $\mathbb R_{> 0}$, respectively.
Let $E$ be a set.
We denote the set of all subsets of $E$ as $2^E$.
A function $d\colon E \times E \to \mathbb R_{\ge 0}$ is called a \emph{metric} (on $E$) if it satisfies the following conditions: for $x, y, z \in E$, (1) $d(x, y) = 0$ if and only if $x = y$; (2) $d(x, y) = d(y, x)$; (3) $d(x, z) \le d(x, y) + d(y, z)$.
Suppose that $E \subseteq \mathbb R^m$ for some integer $m$.
For $x \in E$, we denote by $x_i$ the $i$th component of $x$.
If $d(x, y) = \sum_{1 \le i \le m} |x_i - y_i|$ holds for $x, y \in E$, then $d$ is called an \emph{$\ell_1$-metric}.

Let $E$ be a finite set.
For $X, Y \subseteq E$, the symmetric difference between $X$ and $Y$ is denoted by $X \mathbin{\triangle} Y$ (i.e., $X \mathbin{\triangle} Y = (X \setminus Y) \cup (Y \setminus X)$). 
Let $w: E \to \mathbb R_{>0}$.
A \emph{weighted Hamming distance} is a function $d: 2^E \times 2^E \to \mathbb R_{\ge 0}$ such that for $X, Y \subseteq E$, $d_w(X, Y) = w(X \mathbin{\triangle} Y)$, where $w(Z) = \sum_{x \in Z}w(x)$ for $Z \subseteq E$.
Suppose that $E = \{1, 2, \ldots, m\}$.
We can regard each subset $X \subseteq E$ as an $m$-dimensional vector $x = (x_1, \ldots, x_m)$ defined by $x_i = w(i)$ if $i \in X$ and $x_i = 0$ otherwise, for $1 \le i \le m$.
It is easy to observe that for $X, Y \subseteq E$, $d_w(X, Y) = \sum_{1 \le i \le m}|x_i - y_i|$, where $x$ and $y$ are the vectors corresponding to $X$ and $Y$, respectively.
Thus, the weighted Hamming distance $d_w$ can be considered as an $\ell_1$-metric.

In this paper, we focus on the following diversity measure $\dsum{\cdot}$, called the \emph{sum diversity}.
Let ${\mathcal Y}= \{Y_1,\dots, Y_{k}\}$ be a collection of subsets of $E$ and $w\colon E \to \mathbb R_{\ge 0}$ be a weight function. 
We define $\dsum{\mathcal Y} = \sum_{1 \le i < j \le k} d_w(Y_i, Y_j)$.

Our problem \MSDC{} is defined as follows.

\begin{definition}[\MSDC{}]
    Given a finite set $E$, an integer $k$, a weight function $w\colon E \to \mathbb R_{\ge 0}$, and a membership oracle for $\mathcal X \subseteq 2^E$,
    the task of \MSDC{} is to find a set $\mathcal Y = \set{Y_1, Y_2, \ldots, Y_k}$ of $k$ distinct subsets $Y_1, Y_2, \ldots, Y_k \in \mathcal X$ that maximizes the sum diversity $\dsum{\mathcal Y}$.
\end{definition}

Each set in $\mathcal X$ is called a \emph{feasible solution}.
In \MSDC{}, the set $\mathcal X$ of feasible solutions is not given explicitly, while we can test whether a set $X \subseteq E$ belongs to $\mathcal X$.
Our problem \MSDC{} is highly related to the problem of packing disjoint feasible solutions.
\begin{observation}\label{obs:packing}
    Suppose that all sets in $\mathcal X$ have the same cardinality $r$ and $w(e) = 1$ for $e \in E$.
    Let $Y_1, Y_2, \ldots, Y_k \in \mathcal X$ be $k$ distinct subsets.
    Then, $\dsum{\set{Y_1, \ldots, Y_k}} \ge kr(k - 1)$ if and only if $Y_i \cap Y_j = \emptyset$ for $1 \le i < j \le k$.
\end{observation}
This observation implies several hardness results of \MSDC, which will be discussed in \Cref{sec:appl}.

We particularly focus on the approximability of \MSDC{} for specific sets of feasible solutions.
For a maximization problem, we say that an approximation algorithm has factor $0 < \alpha \le 1$ if given an instance $I$, the algorithm outputs a solution with objective value ${\rm ALG}(I)$ such that ${\rm ALG}(I) / {\rm OPT}(I) \ge \alpha$, where ${\rm OPT}(I)$ is the optimal value for $I$.
A \emph{polynomial-time approximation scheme (PTAS)} is an approximation algorithm that takes an instance $I$ and a constant $\varepsilon > 0$, the algorithm outputs a solution with ${\rm ALG(I) / {\rm OPT}(I)} \ge 1 - \varepsilon$ in polynomial time.

\newcommand{\MSD}{\textsc{Max-Sum Diversification}}
\subsection{A technique for \MSD }\label{sec:ref}

Our framework is based on approximation algorithms for a similar problem \MSD{}.
Let $X$ be a set, and let $d \colon X \times X \to \mathbb R_{\ge 0}$ be a metric.
In what follows, for $Y \subseteq X$, we denote $\sum_{x, y \in Y} d(x, y)$ as $d(Y)$.

\begin{definition}[\MSD{}]
    Given a metric $d\colon X \times X \to \mathbb R_{\ge 0}$ on a finite set $X$ and an integer $k$,
    the task of \MSD{} is to find a subset $Y \subseteq X$ with $|Y| = k$ that maximizes $d(Y)$.
\end{definition}

\MSD{} is studied under various names such as \textsc{MAX-AVG Facility Dispersion} and \textsc{Remote-clique}~\cite{Cevallos:MOR:2019,Ravi:OR:1994}.
\MSD{} is known to be \NP-hard~\cite{Ravi:OR:1994}.
Cevallos et al.~\cite{Cevallos:MOR:2019} devised a PTAS for \MSD.
Their algorithm is based on a rather simple local search technique, but their analysis of the approximation factor and the iteration bound are highly nontrivial. 
Our framework is based on their algorithm, which is briefly sketched below.

\begin{algorithm}[t]
    \caption{A $(1-2/k)$-approximation algorithm for \MSD{}.}
    \label{algo:local}
    \Procedure{\Local{$X, d, k$}}{
        $Y \gets $ arbitrary $k$ elements in $X$\label{algo:init}\;
        \For{$i = 1, \ldots, \lceil \frac{k(k-1)}{(k+1)} \ln (\frac{(k+2)(k-1)^2}{4}) \rceil$}{
            \If{$\exists$ pair $(x, y) \in (X \setminus Y)\times Y$ such that $d(Y - y + x) > d(Y)$}{
                $(x, y) \gets \underset{(x, y)\in (X\setminus Y) \times Y}{\arg\max}d(Y - y + x)$\label{algo:exchange}\;
                $Y \gets Y - y + x$\;
            }
        }
    Output $Y$\;
    }
\end{algorithm}

A pseudocode of the algorithm due to \cite{Cevallos:MOR:2019} is given in \Cref{algo:local}.
In this algorithm, we first pick an arbitrary set of $k$ elements in $X$, which is denoted by $Y \subseteq X$.
Then, we find a pair of elements $x \in X \setminus Y$ and $y \in Y$ that maximizes $d(Y - y + x)$ and update $Y$ by $Y - y + x$ if $d(Y - y + x) > d(Y)$.
We repeat this update procedure $\lceil \frac{k(k-1)}{(k+1)} \ln (\frac{(k+2)(k-1)^2}{4}) \rceil = \order{k \log k}$ times.
Since we can find a pair $(x, y)$ in $\order{\size{X}k\tau}$ time, where $\tau$ is the running time to evaluate the distance function $d(x, y)$ for $x, y \in X$, the following lemma holds.
\begin{lemma}\label{lem:time:local}
    \Cref{algo:local} runs in time $\order{\size{X}k^2 \tau \log k}$.
\end{lemma}

They showed that if the metric $d$ is \emph{negative type}, then the approximation ratio of \Cref{algo:local} is at least $1-2/k$~\cite{Cevallos:MOR:2019}.
The metric $d$ is \emph{negative type} if for any coefficient $b \in \mathbb R^{\size{X}}$ such that $\sum_{b_i \in b}b_i = 0$, $\sum_{x_i \in X}\sum_{x_j \in X}b_ib_jd(x_i, x_j) \le 0$.
It is known that every $\ell_1$-metric is negative type~\cite{Deza:1997,Cevallos:SoCG:2016}.

\begin{theorem}[\cite{Cevallos:MOR:2019}]\label{thm:cevallos}
    If $d\colon X \times X \to \mathbb R_{\ge 0}$ is a negative type metric, then the approximation ratio of \Cref{algo:local} is $1 - 2/k$.
\end{theorem}

They further observed that the above theorem implies that \MSD{} admits a PTAS as follows.
Let $\varepsilon$ be a positive constant.
When $\varepsilon < 2/k$, that is, $k < 2/\varepsilon$, then $k$ is constant. Thus, we can solve \MSD{} in time $\size{X}^{O(1/\varepsilon)}$ by using a brute-force search.
Otherwise, the above $(1-2/k)$-approximation algorithm achieves factor $1-\varepsilon$.
Thus, \MSD{} admits a PTAS, provided that $d$ is a negative type metric.

\begin{corollary}[\cite{Cevallos:MOR:2019}]\label{cor:MSD:ptas}
    If $d\colon X \times X \to \mathbb R_{\ge 0}$ is a negative type metric, then \MSD{} admits a PTAS. 
\end{corollary}

\section{A framework for finding diverse solutions without duplications}\label{sec:framework}
In this section, we propose a framework for designing approximation algorithms for \MSDC{}.
The basic strategy of our framework is the local search algorithm described in the previous section.
Let $E$ be a finite set and let $\mathcal X \subseteq 2^E$ be a set of feasible solutions.
We set $X \coloneqq \mathcal X$ and apply the local search algorithm for \MSD{} to $(X, d_w, k)$.
Recall that our diversity measure $d_{\rm sum}$ is the sum of weighted Hamming distances $d_w$.
Moreover, $d_w$ is an $\ell_1$-metric, as observed in the previous section.
By \Cref{thm:cevallos}, the local search algorithm for \MSD{} has an approximation factor $1 - 2/k$.
However, the running time of a straightforward application of \Cref{lem:time:local} is $O(\size{\mathcal X}\cdot\size{E}k^2\log k)$ even if the feasible solutions in $\mathcal X$ can be enumerated in $\order{\size{\mathcal X}\cdot\size{E}}$ total time, which may be exponential in the input size $\size{E}$.

A main obstacle to applying the local search algorithm is that from a current set $\mathcal Y = \set{Y_1, \ldots, Y_k}$ of feasible solutions, we need to find a pair of feasible solutions $(X, Y) \in (\mathcal X \setminus \mathcal Y) \times \mathcal Y$ maximizing $\dsum{\mathcal Y - Y + X}$.
To overcome this obstacle, we exploit \emph{top-$k$ enumeration algorithms}.
Let $w' \colon E \to \mathbb R$ be a weight function.
An algorithm $\mathcal A$ is called a \emph{top-$k$ enumeration algorithm for $(E, \mathcal X, w', k)$} if for a positive integer $k$, $\mathcal A$ finds $k$ feasible solutions $Y_1, \ldots, Y_k \in \mathcal X$ such that for any $Y \in \set{Y_1, \ldots, Y_k}$ and $X \in \mathcal X \setminus \set{Y_1, \ldots, Y_k}$, $w'(X) \le w'(Y)$ holds.
By using $\mathcal A$, we can compute the pair $(X, Y)$ as follows.

We first guess $Y \in \mathcal Y$ in the pair $(X, Y)$ and let $\mathcal Y' = \mathcal Y \setminus \{Y\}$.
To find the pair $(X, Y)$, it suffices to find $X \in \mathcal X \setminus \mathcal Y'$ that maximizes $\sum_{Y' \in \mathcal Y'} w(X \triangle Y')$.
For an element $e \in E$, we define a new weight $w'(e) \coloneqq w(e)(\mathit{Ex}(e, \mathcal Y') - \mathit{In}(e, \mathcal Y'))$, where $\mathit{In}(e, \mathcal Y')$ (resp. $\mathit{Ex}(e, \mathcal Y')$) is the number of feasible solutions in $\mathcal Y'$ that contain $e$ (resp. do not contain $e$).
For notational convenience, we fix $\mathcal Y'$ and write $\mathit{In}(e)$ and $\mathit{Ex}(e)$ to denote $\mathit{In}(e, \mathcal Y')$ and $\mathit{Ex}(e, \mathcal Y')$, respectively.
The following lemma shows that
a feasible solution $X$ that maximizes $w'(X)$ also maximizes $\sum_{Y' \in \mathcal Y'}w(X \mathbin{\triangle} Y')$.

\begin{lemma}\label{lem:weight}
    For any feasible solution $X \in \mathcal X$, $\sum_{Y' \in \mathcal Y'}w(X \mathbin{\triangle} Y') = w'(X) + \sum_{e \in E}w(e)\cdot\mathit{In}(e)$.
\end{lemma}
\begin{proof}
    The contribution of $e \in X$ to $w(X \mathbin{\triangle} Y')$ is $w(e)$ if $e \not\in Y'$, and $0$ otherwise. 
    Thus, $e \in X$ contributes $w(e)\cdot\mathit{Ex}(e)$ to $\sum_{Y' \in \mathcal Y'}w(X \mathbin{\triangle} Y')$.
    Similarly, $e \in E \setminus X$ contributes $w(e)\cdot \mathit{In}(e)$ to $\sum_{Y' \in \mathcal Y'}w(X \mathbin{\triangle} Y')$.
    This gives us $\sum_{Y' \in \mathcal Y'}w(X \mathbin{\triangle} Y') = w'(X) + \sum_{e \in E}w(e)\cdot\mathit{In}(e)$ as follows.
    \begin{align*}
        &\sum_{Y' \in \mathcal Y'}w(X \mathbin{\triangle} Y') \\
        &=  \sum_{e \in X}w(e)\cdot \mathit{Ex}(e) + \sum_{e \in E \setminus X}w(e)\cdot \mathit{In}(e)\\
        &=  \sum_{e \in X}w(e)\cdot \mathit{Ex}(e) + \sum_{e \in E}w(e)\cdot \mathit{In}(e)  - \sum_{e \in X}w(e)\cdot \mathit{In}(e) \\
        &=  \sum_{e \in X}w(e)(\mathit{Ex}(e) - \mathit{In}(e)) + \sum_{e \in E}w(e)\cdot \mathit{In}(e) \\
        &=  w'(X) + \sum_{e \in E}w(e)\cdot \mathit{In}(e). \qed
    \end{align*}
\end{proof}

From the above lemma, we can find the pair $(X, Y)$ with a top-$k$ enumeration algorithm $\mathcal A$ for $(E, \mathcal X, w', k)$ as follows.
By~\Cref{lem:weight}, for any feasible solution $X \in \mathcal X$, $\sum_{Y' \in \mathcal Y'}w(X \mathbin{\triangle} Y') = w'(X) + \sum_{e \in E} w(e) \cdot \mathit{In}(e)$.
Since the second term does not depend on $X$, to find a feasible solution $X$ maximizing the left-hand side, it suffices to maximize $w'(X)$ subject to $X \in \mathcal X \setminus \mathcal Y'$.
The algorithm $\mathcal A$ allows us to find $k$ feasible solutions $Z_1, \ldots, Z_k$ such that $w'(Z_1) \ge \cdots \ge w'(Z_k) \ge w'(Z)$ for any feasible solution $Z$ other than $Z_1, \ldots, Z_k$.
As $|\mathcal Y'| < k$, at least one of these solutions provides such a solution $X$.

The entire algorithm is as follows.
We first find a set of $k$ distinct feasible solutions in $\mathcal X$ using the enumeration algorithm $\mathcal A$.
Then, we repeat the local update procedure described above $O(k \log k)$ times.
Suppose that $\mathcal A$ enumerates $k$ feasible solutions in time $O((|E|+k)^c)$ for some constant $c$.
Then, the entire algorithm runs in time $O((|E| + k)^c|E|k^2\log k)$ as we can compute the pair $(X, Y)$ in time $O((|E|+k)^c|E|k)$ by simply guessing $Y \in \mathcal Y$.

The approximation factor $1 - 2/k$ does not give a reasonable bound for $k = 2$.
In this case, however, we still have an approximation factor $1 / 2$ with a greedy algorithm for \MSD~\cite{Benjamin:AL:2009}, which is described as follows.
Initially, we set $\mathcal Y = \{Y_1\}$ with arbitrary $Y_1 \in \mathcal X$.
Then, we compute a feasible solution $Y_2 \in \mathcal X \setminus \mathcal Y$ maximizing $\sum_{Y \in \mathcal Y} w(Y_2 \triangle Y)$.
By~\Cref{lem:weight} and the above discussion, we can find such a solution $Y_2$ with a top-$k$ enumeration algorithm for $(E, \mathcal X, w', k)$, where $w'(e) \coloneqq w(e) \cdot (\mathit{Ex}(e, \mathcal Y) - \mathit{In}(e, \mathcal Y))$ for $e \in E$.
We repeat this $k - 1$ times so that $\mathcal Y$ contains $k$ feasible solutions.
As discussed in this section, the approximation factor of this algorithm is $1 / 2$ as in~\cite{Benjamin:AL:2009}.
Thus, the following theorem holds.

\begin{theorem}\label{thm:main-framework}
    Let $E$ be a finite set, $\mathcal X \subseteq 2^E$, and $w\colon E \to \mathbb R_{> 0}$.
    Suppose that there is  a top-$k$ enumeration algorithm for $(E, \mathcal X, w', k)$ that runs in $\order{(\size{E} + k)^c}$ time, where 
    $w'\colon E\to \mathbb R$ is an arbitrary weight function.
    Then, there is a $\max(1 - 2/k, 1/2)$-approximation algorithm for \MSDC{} that runs in $\order{(\size{E} + k)^{c} |E| k^2\log k}$ time.
    Moreover, if there is a polynomial-time exact algorithm for \MSDC{} for constant $k$, then it admits a PTAS.
\end{theorem}

We note that the approximation factor of the framework of \Cref{thm:main-framework} is gradually improved when the number $k$ of solutions increases.
However, the dependency of the running time on $k$ is only polynomial, which allows us to find a moderate number of diverse solutions efficiently.

\section{Applications of the framework}\label{sec:appl}
To complete the description of approximation algorithms based on our framework, we need to develop top-$k$ enumeration algorithms for specific problems.
In this section, we design top-$k$ enumeration algorithms for matchings, common bases of two matroids, and interval schedulings with cardinality $r$.

Our top-$k$ enumeration algorithms are based on a well-known technique used in~\cite{Lawler:MS:1972} (also discussed in \cite{Eppstein:EA:2008}).
The key to enumeration algorithms is the following \textsc{Weighted Extension}.

\begin{definition}[\textsc{Weighted Extension}]
    Given a finite set $E$, a set of feasible solutions $\mathcal X \subseteq 2^E$ as a membership oracle, a weight function $w'\colon E \to \mathbb R$, and a pair of disjoint subsets $\mathit{In}$ and $\mathit{Ex}$ of $E$,
    the task is to find a feasible solution $X \in \mathcal X$ that satisfies $\mathit{In} \subseteq X$ and $X\cap \mathit{Ex} = \emptyset$ maximizing $w'(X)$ subject to these conditions.
\end{definition}

If we can solve the above problem in $\order{\size{E}^c}$ time, then we obtain a top-$k$ enumeration algorithm for $(E, \mathcal X, w', k)$ that runs in $\order{k\size{E}^{c+1}}$ time.

\begin{lemma}[\cite{Lawler:MS:1972}]\label{lem:topk}
    Suppose that \textsc{Weighted Extension} for $(E, \mathcal X, w', k)$ can be solved in $\order{\size{E}^c}$ time. 
    Then, there is an $\order{k\size{E}^{c+1}}$-time top-$k$ enumeration algorithm for $(E, \mathcal X, w', k)$.
\end{lemma}

\subsection{Matchings}
Matching is one of the most fundamental combinatorial objects in graphs, and the polynomial-time algorithm for computing a maximum weight matching due to \cite{Edmonds:paths:1965} is a cornerstone result in this context.
Finding diverse matchings has also been studied so far~\cite{Hanaka:AAAI:2021,Hanaka:AAAI:2022,Fomin:ISAAC:2020,Fomin:STACS:2021}.
Let $G = (V, E)$ be a graph. 
A set of edges $M$ is a \emph{matching} of $G$ if $M$ has no pair of edges that share a common endpoint. 
A matching $M$ is called a \emph{perfect matching} of $G$ if every vertex in $G$ is incident to an edge in $M$.
By using our framework, we design an approximation algorithm for finding diverse matchings.
The formal definition of the problem is as follows.
\newcommand{\MSDM}{\textsc{Diverse Matchings}}
\begin{definition}[\MSDM{}]
    Given a graph $G = (V, E)$, a weight function $w\colon E\to \mathbb R_{> 0}$, and integers $k$ and $r$,
    the task of \MSDM{} is to find $k$ distinct matchings $M_1, \ldots, M_k$ of cardinality at least $r$ that maximize $\dsum{\{M_1, \ldots, M_k\}}$.
\end{definition}

To apply our framework, it suffices to show that \textsc{Weighted Extension} for matchings can be solved in polynomial time.
Our method is similar to a reduction from the maximum weight perfect matching problem to the maximum weight matching problem~\cite{10.1145/2529989}.
Let $\mathit{In}, \mathit{Ex} \subseteq E$ be disjoint subsets of edges and let $w'\colon E \to \mathbb R$.
Then, our goal is to find a matching $M$ of $G$ with $|M| \ge r$ such that $\textit{In} \subseteq M$ and $\textit{Ex} \cap M = \emptyset$, and $M$ maximizes $w'(M)$ subject to these constraints.
By guessing the cardinality of $M$, it suffices to find such a matching $M$ with cardinality exactly $r$.
This problem can be reduced to that of finding a maximum weight perfect matching as follows.
We assume that $\textit{In}$ is a matching of $G$ as otherwise there is no matching containing it. 
Let $G' = (V', E')$ be the graph obtained from $G$ by removing (1) all edges in $\textit{Ex}$ and (2) all end vertices of edges in $\textit{In}$.
Then, it is easy to see that $M$ is a matching of $G$ with $\textit{In} \subseteq M$ and $\textit{Ex} \cap M = \emptyset$ if and only if $M \setminus \textit{In}$ is a matching of $G'$.
Thus, it suffices to find a maximum weight matching of cardinality exactly $r' = r - \size{\textit{In}}$ in $G'$.
To this end, we add $|V'| - 2r'$ vertices $U$ to $G'$ and add all possible edges between vertices $v \in V'$ and $u \in U$. 
The graph obtained in this way is denoted by $H = (V' \cup U, E \cup F)$, where $F = \inset{\set{u, v}}{u \in U, v \in V'}$. 
We extend the weight function $w'$ by setting $w'(f) = 0$ for $f \in F$. 
Then, the following lemma holds.

\begin{lemma}
    Let $M^*$ be a maximum weight perfect matching in $H$.
    Then, $M^* \setminus F$ is a matching of cardinality $r'$ in $G'$ such that for every cardinality-$r'$ matching $M'$ in $G'$, it holds that $w'(M') \le w'(M^* \setminus F)$.
\end{lemma}
\begin{proof}
    Since $M^*$ is a perfect matching and any edge incident to $U$ is contained in $F$, $M^*$ must contain exactly $|U|$ edges of $F$.
    This implies that the perfect matching $M^*$ contains exactly $r'$ edges of $G'$.
    Suppose that there is a cardinality-$r'$ matching $M'$ in $G'$ such that $w'(M') > w'(M^* \setminus F)$.
    As every vertex in $U$ is adjacent to $V'$, we can choose exactly a set $N \subseteq F$ of $|U|$ edges between $U$ and $V'$ so that $M' \cup N$ forms a perfect matching in $H$.
    Then, we have $w'(M' \cup N) = w'(M') + w'(N) > w'(M^* \setminus F) + w'(M^* \cap F) = w'(M^*)$ as $w'(N) = w'(M^* \cap F) = 0$, contradicting the fact that $M^*$ is a maximum weight perfect matching of $H$.
%
\end{proof}

Thus, \textsc{Weighted Extension} for a matching of cardinality at least $r$ is solvable in polynomial time using a maximum weight matching algorithm~\cite{Edmonds:paths:1965}.
By~\Cref{thm:main-framework} and \Cref{lem:topk}, we have the following theorem.

\begin{theorem}
    There is a polynomial-time approximation algorithm for \MSDM{} with approximation factor $\max(1 - 2/k, 1/2)$.
\end{theorem}

\subsection{Common bases of two matroids}
Let $E$ be a finite set and let a non-empty family of subsets $\mathcal I$ of $E$.
The pair $\mathcal M = (E, \mathcal I)$ is a \emph{matroid} if 
(1) for each $X \in \mathcal I$, every subset of $X$ is included in $\mathcal I$ and (2) if $X, Y \in \mathcal I$ and $\size{X} < \size{Y}$, then there exists an element $e \in Y \setminus X$ such that $X \cup \set{e} \in \mathcal I$.
Each set in $\mathcal I$ is called an \emph{independent set} of $\mathcal M$.
An inclusion-wise maximal independent set $I$ of $\mathcal M$ is a \emph{base} of $\mathcal M$.
Because of condition (2), all bases in $\mathcal M$ have the same cardinality.
For two matroids $\mathcal M_1 = (E, \mathcal I_1)$ and $\mathcal M_2 = (E, \mathcal I_2)$, 
a subset $X \subseteq E$ is a \emph{common base of $\mathcal M_1$ and $\mathcal M_2$} if $X$ is a base of both $\mathcal M_1$ and $\mathcal M_2$.
In this subsection, we give an approximation algorithm for diverse common bases of two matroids.

\newcommand{\DMCB}{\textsc{Diverse Matroid Common Bases}}

\begin{definition}[\DMCB]
    Given matroids $\mathcal M_1 = (E, \mathcal I_1)$ and $\mathcal M_2 = (E, \mathcal I_2)$ as membership oracles, a weight function $w\colon E \to \mathbb R_{> 0}$, and an integer $k$,
    the task of \DMCB\ is to find $k$ distinct common bases $B_1, \ldots, B_k$ of $\mathcal M_1$ and $\mathcal M_2$
    that maximize $\dsum{B_1, \ldots, B_k}$.
\end{definition}

Given two matroids $\mathcal M_1 =(E,\mathcal I_1)$ and $\mathcal M_2 =(E,\mathcal I_2)$ as membership oracles, the problem of partitioning $E$ into $k$ common bases of $\mathcal M_1$ and $\mathcal M_2$ is a notoriously hard problem, which requires an exponential number of membership queries~\cite{Berczi:MP:2021}.
This fact together with~\Cref{obs:packing} implies that \DMCB{} cannot be solved with polynomial number of membership queries in our problem setting.
Given this fact, we develop a constant-factor approximation algorithm for \DMCB{}.
To this end, we show that \textsc{Weighted Extension} for common bases of two matroids can be solved in polynomial time.

Similarly to the case of matchings, we can find a maximum weight common base $B \in \mathcal I_1 \cap \mathcal I_2$ subject to $\mathit{In} \subseteq B$ and $\mathit{Ex} \cap B = \emptyset$ for given disjoint $\mathit{In}, \mathit{Ex} \subseteq E$, which is as follows.
Let $\mathcal M = (E, \mathcal I)$ be a matroid.
For $X \subseteq E$, we let $\mathcal M \setminus X = (E \setminus X, \mathcal J)$, where $\mathcal J = \inset{J \setminus X}{J \in \mathcal I}$.
Then, $\mathcal M \setminus X$ is a matroid (see~\cite{Oxley:book:2006}).
Similarly, for $X \subseteq E$, we let $\mathcal M \slash X = (E \setminus X, \mathcal J')$, where $\mathcal J' = \inset{J}{J \cup X \in \mathcal I, J \subseteq E \setminus X}$.
Then $(E, \mathcal J)$ is also a matroid (see~\cite{Oxley:book:2006}).
For two matroids $\mathcal M_1$ and $\mathcal M_2$, we consider two matroids $\mathcal M'_1 = (\mathcal M_1 \setminus \mathit{Ex}) \slash \mathit{In}$ and $\mathcal M'_2 = (\mathcal M_2 \setminus \mathit{Ex}) \slash \mathit{In}$.
For every independent set $X$ in $\mathcal M'_1$ and $\mathcal M'_2$, 
$X$ does not contain any element in $\mathit{Ex}$ and $X \cup \mathit{In}$ is an independent set in both $\mathcal M_1$ and $\mathcal M_2$.
Thus, \textsc{Weighted Extension} can be solved by computing a maximum weight common base in $\mathcal M'_1$ and $\mathcal M'_2$, which can be solved in polynomial time (see Theorem~41.7 in \cite{Schrijver:book}).
By~\Cref{thm:main-framework} and \Cref{lem:topk}, the following theorem holds.

\begin{theorem}
    There is a polynomial-time approximation algorithm for \DMCB{} with approximation factor $\max(1 - 2/k, 1/2)$, provided that the membership oracles for $\mathcal M_1$ and $\mathcal M_2$ can be evaluated in polynomial time.
\end{theorem}

\subsection{Minimum cuts}\label{subsec:mincut}
\newcommand{\DMGC}{\textsc{Diverse Minimum Cuts}}

Let $G = (V, E)$ be a graph. 
A partition of $V$ into two non-empty sets $A$ and $B$ is called a \emph{cut} of $G$.
For a cut $(A, B)$ of $G$, the set of edges having one end in $A$ and the other end in $B$ is denoted by $E(A, B)$.
When no confusion arises, we may refer to $E(A, B)$ as a cut of $G$.
The \emph{size} of a cut $C = E(A, B)$ is defined by $|E(A, B)|$.
A cut $C$ is called a \emph{minimum cut} of $G$ if there is no cut $C'$ of $G$ with $|C'| < |C|$.
In this section, we consider the following problem.

\begin{definition}[\DMGC]
    Given a graph $G = (V, E)$ with an edge-weight function $w\colon E \to \mathbb R_{\ge 0}$ and an integer $k$, the task of \DMGC{} is to find $k$ distinct minimum cuts $C_1, \ldots, C_k \subseteq E$ of $G$ that maximize $\dsum{\set{C_1, \ldots, C_k}}$.
\end{definition}

An important observation for this problem is that the number of minimum cuts of any graph $G$ is $\order{|V|^2}$~\cite{Karger:ACM:2000}.
Moreover, we can enumerate all minimum cuts in a graph in polynomial time~\cite{Yeh:Algorithmica:2010,Yannakakis:ICALP:92}. 
Thus, we can solve both \textsc{Weighted Extension} for minimum cuts and \DMGC{} for constant $k$ in polynomial time, yielding a PTAS for \DMGC{}.

\begin{theorem}
    \DMGC{} admits a PTAS.
\end{theorem}

Given this, it is natural to ask whether \DMGC{} admits a polynomial-time algorithm.
However, we show that \DMGC{} is \NP-hard even if $G$ has a cut of size $3$.
Let $\lambda(G)$ be the size of a minimum cut of $G$.

\begin{theorem}\label{thm:cut-hardness}
    \DMGC{} is \NP-hard even if $\lambda(G) = 3$ and $w(e) = 1$ for every edge.
\end{theorem}
\begin{proof}
    We give a reduction from the maximum independent set problem on cubic graphs, which is known to be \NP-complete~\cite{FLEISCHNER20102742}.
    For a graph $H$, we denote by $\alpha(H)$ the maximum size of an independent set of $H$.
    Let $(H, k)$ be an instance of this problem, where every vertex of $H$ has degree exactly $3$.

    Let $H'$ be the graph obtained from $H$ by subdividing each edge twice, that is, each edge $e = \set{u, v} \in E(H)$ is replaced by a path $(u, a_e, b_e, v)$ of three edges, and let $D = \set{a_e, b_e \mid e \in E(H)}$ be the set of new vertices.
    We construct a graph $G = (V, E)$ from $H'$ by adding a new vertex $v^*$ and adding an edge between $v^*$ and each vertex in $D$, and we set $w(e) = 1$ for every $e \in E$.
 Then, every vertex in $V \setminus \set{v^*} = V(H')$ has degree exactly three in $G$.
    Finally, we set $k' = k + |E(H)|$, and we show that $G$ has $k'$ distinct minimum cuts $C_1, \ldots, C_{k'} \subseteq E$ with $\dsum{\set{C_1, \ldots, C_{k'}}} \ge 3k'(k'-1)$ if and only if $\alpha(H) \ge k$.
    From the construction of $H'$, it is easy to show that $\alpha(H') \ge k + |E(H)|$ if and only if $\alpha(H) \ge k$.

    We show that every cut of $G$ of size at most three is a trivial cut $E_G(\set{v}, V \setminus \set{v})$ for some $v \in V(H')$.
    Let $C = E(X, V \setminus X)$ be a cut of $G$ with $|C| \le 3$ and, without loss of generality, $v^* \in V \setminus X$.
    Since every vertex of $X$ has degree exactly three in $G$ and since every edge of $G$ joining two vertices of $X$ is an edge of $H'$, we have $|C| = 3|X| - 2|E(H'[X])|$.
    Suppose that $H'[X]$ contains a cycle.
    Every cycle of $H'$ is obtained from a cycle of $H$ by subdividing each of its edges twice, and a cycle of the simple graph $H$ has at least three edges; hence, every cycle of $H'$ contains at least six vertices of $D$.
    Each vertex of $X \cap D$ is adjacent to $v^* \in V \setminus X$, so $|C| \ge |X \cap D| \ge 6$, a contradiction.
    Therefore, $H'[X]$ is a forest, and $|E(H'[X])| \le |X| - 1$ yields $|C| \ge 3|X| - 2(|X| - 1) = |X| + 2$.
    Together with $|C| \le 3$, we obtain $|X| = 1$, that is, $C = E_G(\set{v}, V \setminus \set{v})$ for some $v \in V(H')$, and $|C| = 3$.
    In particular, every cut of $G$ has size at least three, and hence $\lambda(G) = 3$ and the minimum cuts of $G$ are exactly the trivial cuts of the vertices in $V(H')$, which are pairwise distinct; since $|V(H')| = |V(H)| + 2|E(H)| \ge k'$, the graph $G$ always has $k'$ distinct minimum cuts.

    We show that $G$ has $k'$ distinct minimum cuts $C_1, \ldots, C_{k'}$ with $\dsum{\set{C_1, \ldots, C_{k'}}} \ge 3k'(k'-1)$ if and only if $\alpha(H') \ge k'$.
    Since every minimum cut of $G$ consists of exactly three edges, any two minimum cuts $C_i$ and $C_j$ satisfy $d_w(C_i, C_j) = |C_i \bigtriangleup C_j| = 6 - 2|C_i \cap C_j| \le 6$, with equality if and only if $C_i \cap C_j = \emptyset$; hence $\dsum{\set{C_1, \ldots, C_{k'}}} \le 3k'(k'-1)$, with equality if and only if the cuts are pairwise disjoint.
    Since every edge of $G$ not in $H'$ is incident to $v^*$, two trivial cuts $E_G(\set{u}, V \setminus \set{u})$ and $E_G(\set{v}, V \setminus \set{v})$ with $u, v \in V(H')$ are disjoint if and only if $u$ and $v$ are not adjacent in $H'$.

\end{proof}

When $\lambda(G) = 1$, then \DMGC{} is trivially solvable in linear time as the problem can be reduced to finding all the bridges in $G$.
We show that \DMGC{} can be solved in polynomial time when $\lambda(G) \le 2$.
We reduce the problem to that of finding a subgraph of prescribed size with maximizing the sum of convex functions on their degrees of vertices.

\begin{theorem}[\cite{Apollonio:Minconvex:2009}]\label{thm:deg-conv}
    Given an undirected graph $H$, an integer $k$, and convex functions $f_v: \mathbb N_{\ge 0} \to \mathbb R$ for $v \in V(H)$, the problem of finding $k$-edge subgraph $H'$ of $H$ maximizing $\sum_{v \in V(H)} f_v(d_{H'}(v))$ is solvable in polynomial time, where $d_{H'}(v)$ is the degree of $v$ in $H'$.
\end{theorem}

We first enumerate all minimum cuts of $G$ in polynomial time.
If $G$ has no $k$ minimum cuts, then the instance is trivially infeasible.
Suppose otherwise.
We construct a graph $H$ whose vertex set corresponds to $E$, and the edge set of $H$ is defined as follows.
For each pair $e, f \in E$, we add an edge between $e$ and $f$ to $H$ if $\{e, f\}$ is a cut of $G$.
Obviously, the graph $H$ can be constructed in polynomial time.
For each $e \in E$, we let $f_e(i) \coloneqq w(e) \cdot i \cdot (k - i)$ for $0 \le i \le k$ and $f_e(i) = \infty$ for $i > k$.
Clearly, the function $f_e$ is convex.
Let $C_1, \ldots, C_k \subseteq E$ be $k$ minimum cuts of $G$.
For each $e$, we denote by $m(e)$ the number of occurrences of $e$ among $C_1, \ldots, C_k$. 
Since each edge in $E$ contributes $w(e) \cdot m(e) \cdot (k - m(e))$ to $\dsum{\set{C_1, \ldots, C_k}}$, we immediately have the following lemma.

\begin{lemma}\label{lem:cut-reduction}
    $H$ has a subgraph $H'$ of $k$ edges such that $\sum_{v \in V(H)} f_e(d_{H'}(e)) \ge t$ if and only if there are $k$ edge cuts $C_1, \ldots, C_k \subseteq E$ of $G$ with $|C_i| = 2$ for $1 \le i \le k$ such that $\dsum{\set{C_1, \ldots, C_k}} \ge t$.
\end{lemma}

By~\Cref{lem:cut-reduction} and \Cref{thm:deg-conv}, \DMGC{} can be solved in $\size{V}^{O(1)}$ time, proving the following theorem.

\begin{theorem}\label{thm:cut:lambda-2}
    \DMGC{} can be solved in polynomial time, provided that $\lambda(G) \le 2$.
\end{theorem}

\subsection{Interval schedulings}\label{sec:is}
For a pair of integers $a$ and $b$ with $a \le b$, the set of all numbers between $a$ and $b$ is denoted by $[a, b]$.
We call $I = [a, b]$ an \emph{interval}. 
For a pair of intervals $I = [a, b]$ and $J = [c, d]$, we say that $I$ \emph{overlaps} $J$ if $I \cap J\neq \emptyset$.
For a set of intervals $\mathcal S = \set{I_1, \ldots, I_r}$, we say that $\mathcal S$ is a \emph{valid scheduling} (or simply a \emph{scheduling}) if for any pair of intervals $I_i, I_j \in \mathcal S$, $I_i$ does not overlap $I_j$.
In particular, we call $\mathcal S$ an \emph{$r$-scheduling} if $\size{\mathcal S} = r$ for $r \in \mathbb N$.
In this section, we deal with the following problem.

\newcommand{\DIS}{{\textsc{Diverse Interval Schedulings}}}

\begin{definition}[\DIS{}]
    Given a set of intervals $\mathcal I$, a weight function $w\colon \mathcal I \to \mathbb R_{> 0}$, and integers $k$ and $r$, the task of \DIS{} is to find $k$ distinct $r$-schedulings $\mathcal S_1, \ldots, \mathcal S_k \subseteq \mathcal I$ that maximize $\dsum{\set{\mathcal S_1, \ldots, \mathcal S_k}}$.
\end{definition}

Since the problem of partitioning a set of intervals $\mathcal I  = \set{I_1, \ldots, I_n}$ into $k$ scheduling $\mathcal S_1, \ldots, \mathcal S_k$ such that each $\mathcal S_i$ has exactly $r$ intervals is \NP-hard~\cite{BodlaenderJ:TCS:1995,Gardi:DAM:2009}
\footnote{Note that the \NP-hardness is proven for the case that each $\mathcal S_i$ has \emph{at most} $r$ intervals, but a simple reduction proves the \NP-hardness of this variant.}, by~\Cref{obs:packing}, the following theorem holds.

\begin{theorem}
    \DIS{} is \NP-hard.
\end{theorem}

To apply \Cref{thm:main-framework} to \DIS{}, it suffices to give a polynomial-time algorithm for {\sc Weighted Extension} for interval schedulings.
Observe that if $\mathit{In}$ is not a scheduling, then there is no scheduling containing $\mathit{In}$.
Observe also that we can remove all intervals included in $\mathit{Ex}$ or overlapping some interval in $\mathit{In}$.
Thus, the problem can be reduced to the one for finding a maximum weight scheduling with cardinality $r' = r - \size{\mathit{In}}$.
This problem can be solved in polynomial time by using a simple dynamic programming approach. 

\begin{lemma}\label{lem:max-is-dp}
    Given a set $\mathcal I$ and $w'\colon \mathcal I \to \mathbb R$ and $r' \in \mathbb N$, there is a polynomial-time algorithm finding a maximum weight $r'$-scheduling in $\order{\size{\mathcal I}^2r'}$ time.
\end{lemma}
\begin{proof}
    The algorithm is analogous to that to find a maximum weight independent set on interval graphs, which is roughly sketched as follows.
    We assume that $\mathcal I = \{I_1, I_2, \ldots, I_n\}$ is sorted with respect to their right end points.
    We define ${\rm opt}(p, q)$ as the maximum total weight of a $q$-scheduling $\mathcal S$ in $\{I_1, I_2, \ldots, I_p\}$ such that $I_p \in \mathcal S$ for $0 \le p \le n$ and $0 \le q \le r'$.
    Then, the values of ${\rm opt}(p, q)$ for all $p$ and $q$ can be computed by a standard dynamic programming algorithm in time $O(\size{\mathcal I}^2r')$.
\end{proof}

By \Cref{thm:main-framework} and \Cref{lem:topk}, we obtain a polynomial-time approximation algorithm for \DIS{} with factor $\max(1 - 2/k, 1/2)$.

Finally, we show that \DIS{} can be solved in polynomial time for fixed $k$ using a dynamic programming approach, which implies a PTAS for \DIS{}.

Similarly to the proof of \Cref{lem:max-is-dp}, assume that $\mathcal I = \{I_1, I_2, \ldots, I_n\}$ is sorted with respect to their right end points.
Let $[k] = \{1, 2, \ldots, k\}$.
For each $0 \le p \le \size{\mathcal I}$, we consider a tuple $T = (p, L, R, \Gamma)$, where $L$ and $R$ are vectors in $([n] \cup \{0\})^k$ and $([r] \cup \{0\})^k$, respectively, and $\Gamma$ is a subset of $\binom{[k]}{2}$.
Clearly, the number of tuples is $O(n(n+1)^k(r+1)^k2^{\binom{k}{2}})$, which is polynomial when $k$ is a constant.
We denote by $\ell_i$ and $r_i$ the $i$th component of $L$ and $R$, respectively.
For a tuple $T = (p, L, R, \Gamma)$,
the value ${\rm opt}(T)$ is the maximum value of $\dsum{\set{\mathcal S_1, \ldots, \mathcal S_k}}$ 
for $k$ schedulings under the following four conditions: (1) the maximum index of an interval in $\bigcup_{1\le i\le k} \mathcal S_i$ is $p$ ($p = 0$ if $\bigcup_{1\le i\le k} \mathcal S_i = \emptyset$); (2) for $1 \le i \le k$, the maximum index of an interval in $\mathcal S_i$ is $\ell_i$ ($\ell_i = 0$ if $\mathcal S_i = \emptyset$); (3) for $1 \le i \le k$, $\size{\mathcal S_i} = r_i$; and (4) for $1 \le i < j \le k$, $\set{i, j} \in \Gamma$ if and only if $\mathcal S_i$ and $\mathcal S_j$ are distinct.

We define ${\rm opt}(T) = -\infty$ if no such a set of schedulings exists.
When $R = 
(r, r, \ldots, r)$ and $\Gamma = \binom{[k]}{2}$, 
there is a set of $k$ distinct $r$-schedulings that have the sum diversity ${\rm opt}(T)$ unless ${\rm opt}(T) = -\infty$.
For a tuple $T$, we say that a set of $k$ schedulings is \emph{valid for $T$} if it satisfies the above four conditions.
Hence, among the tuples of the form $(p, L, R, \Gamma)$ with $R = (r, \ldots, r)$ and $\Gamma = \binom{[k]}{2}$, ${\rm opt}(T)$ is the optimal value for \DIS{}.
We next explain the outline of our dynamic programming algorithm to compute ${\rm opt}(T)$ for any $T$.

As a base case, $p = 0$, $L = (0, \ldots, 0)$, $R = (0, \ldots, 0)$, and $\Gamma = \emptyset$ if and only if ${\rm opt}(T) = 0$.
Let $T'$ be a tuple $(p', L', R', \Gamma')$ that satisfies the following conditions:
(1)$p' < p$; (2) for any $1 \le i \le k$, $\ell'_i \le \ell_i$ and $r'_i \le r_i$; and (3) $\Gamma' \subseteq \Gamma$.
We say that a tuple $T'$ satisfying the above conditions is \emph{dominated by $T$}.
We denote the set of tuples dominated by $T$ by $D(T)$.
Let $C(T) = \inset{i}{\ell_i = p}$.
A tuple $T'$ is \emph{valid for $T$} if $T'$ satisfies the following conditions:
(1) $T' \in D(T)$;
(2) if $i \in C(T)$ and $\ell_i > 0$, then interval $I_{\ell_i}$ does not overlap with $I_p$;
(3) if $i \in C(T)$, $r'_i = r_i - 1$, otherwise, $r'_i = r_i$; and
(4) $\Gamma = \Gamma' \cup P(T)$ with $P(T) \coloneqq \inset{\set{i, j} \in \binom{[k]}{2}}{\size{\set{i, j} \cap C(T)} = 1}$.
We denote the set of valid tuples for $T$ as $V(T)$.
We compute ${\rm opt}(T)$ using the following lemma.

\begin{lemma}
    For a tuple $T$,
    \begin{align*}
        {\rm opt}(T) = \underset{T' \in V(T)}{\max}({\rm opt}(T') + w(I_p)\cdot \size{C(T)}\cdot(k - \size{C(T)})).
    \end{align*}
\end{lemma}
\begin{proof}
    Let $T = (p, L, R, \Gamma)$.
    Let $\mathcal S = \set{\mathcal S_1, \ldots, \mathcal S_k}$ be a valid set of schedulings with $\dsum{\set{\mathcal S_1, \ldots, \mathcal S_j}} = {\rm opt}(T)$.
    Then $\mathcal S' = (\mathcal S_1 \setminus \set{I_p}, \ldots, \mathcal S_k \setminus \set{I_p})$ is a valid set of schedulings for $T' \in V(T)$.
    Moreover, $\dsum{\mathcal S} = \dsum{\mathcal S'} + w(I_p) \cdot |C(T)| \cdot (k - |C(T)|)$ as $I_p$ contributes $ w(I_p) \cdot |C(T)| \cdot (k - |C(T)|)$ to the diversity.
    Thus, the left-hand side is at most the right-hand side.
    
    Conversely, let $T'$ be a tuple maximizing the left-hand side and let $\mathcal S' = \set{\mathcal S'_1, \ldots, \mathcal S'_k}$ be a valid set of schedulings for $T'$.
    For each $1 \le i \le k$, we set $\mathcal S_i = \mathcal S'_i \cup \{I_p\}$ if $i \in C(T)$ and $\mathcal S_i = \mathcal S'_i$ otherwise.
    By condition (2) in the definition of a valid tuple, each interval in $\mathcal S'_i$ does not overlap with $I_p$, which means that $\mathcal S_i$ is a scheduling.
    Thus, the right-hand side is at most the left-hand side.
    %
%
\end{proof}

Thus, we can compute ${\rm opt}(T)$ for any $T$ in polynomial time when $k$ is a constant. 
Moreover, from ${\rm opt}(T)$, we can find $k$ schedulings with the maximum sum diversity by a standard trace-back technique.
Combining the approximation algorithm and the above algorithm, we obtain a PTAS.

\begin{theorem}
    \DIS{} admits a PTAS. 
\end{theorem}

It is not hard to see that the above algorithm is modified into the one finding a set $\mathcal S$ of $k$ valid schedulings $\mathcal S_i$ with $|\mathcal S_i| \ge r$ maximizing $\dsum{\set{\mathcal S_1, \ldots, \mathcal S_k}}$.
The modified algorithm simply takes ``at least $r$ intervals'' instead of ``exactly $r$ intervals'', which runs in polynomial time as well.

\newcommand{\opt}[1]{{\rm opt}_{#1}}

\section{An improved framework when allowing duplication}\label{sec:random}
We gave a framework for finding diverse solutions without duplications with an approximation factor $\max(1 -2/k, 1/2)$ if we have a top-$k$ enumeration algorithm.
We improve an approximation factor from $\max(1 -2/k, 1/2)$ to $(1 - \varepsilon)(1 - 1/k)$ for any constant $\varepsilon > 0$ if we allow duplications.
By setting $\varepsilon \le 1/(k-1)$, 
this approximation ratio is better than $1 - 2/k$ since
$(1 - \varepsilon)(1 - 1/k) \le 1 - 2/k$ holds.

In this framework, we assume that we have an algorithm for finding a maximum-weight solution in $\mathcal X$.
It is worth noting that the assumption in this framework is weaker than the assumption required in the framework in \Cref{sec:framework}.
Hence, this framework can be similarly applied to the concrete examples in \Cref{sec:appl}, which were treated as applications of the framework in \Cref{sec:framework}.

We first give a ``randomized'' framework for finding diverse solutions with duplications with an approximation factor $(1 - \varepsilon)(1 - 1/k)$ in expectation.
In a computational model considered in our randomized framework,
we identify a discrete probability distribution $\mathcal D$ over $\mathcal X$ with the set of pairs consisting of elements in 
$\mathcal X$ that have non-zero probability and their corresponding probabilities.
Given a discrete probability distribution $\mathcal D$ over $\mathcal X$,
one can perform sample from $\mathcal D$ in a constant time.
We denote  $X \sim \mathcal D$ that $X$ is drawn from the probability distribution $\mathcal D$ over $\mathcal X$.
We denote by $(X, X') \sim D \times D'$ the process of independently sampling two elements from $\mathcal X$, where $X \sim D$ and $X' \sim D'$.

We first define $\opt{k}$ and $\opt{\infty}$ as follows:
\begin{align*}
    \opt{k} &\coloneq \max_{(X_1,\dots, X_k)\in \mathcal{X}^k} \left(\frac{2}{k(k-1)}\sum_{1\leq i<j\leq k} d(X_i,X_j) \right) \\
    \opt{\infty} &\coloneq \max_{\mathcal{D}\in \Pi(\mathcal{X})} \left(\E_{(X,X')\sim \mathcal{D}\times \mathcal{D}}\left[d(X,X')\right]\right), 
\end{align*}
where $\mathcal{X}$ is the family of all feasible solutions, $\Pi(\mathcal{X})$ is the set of all probability distributions over $\mathcal{X}$, and $d$ is a weighted Hamming distance over $\mathcal{X}\times \mathcal X$.
Note that while we previously defined $\rm opt$ as the sum of pairwise distances, 
in this section it is defined as the average of the pairwise distances.
This modification only multiplies the objective value by 
$2/k(k-1)$, and therefore does not affect an approximation ratio analysis.

The following lemma is essential to analyze an approximation ratio of our framework.

\begin{lemma}\label{lem:kandinf}
It holds that
\begin{align*}
    \frac{k-1}{k}\opt{k}\leq \opt{\infty} \leq \opt{k}.
\end{align*}
\end{lemma}
\begin{proof}
Let $(Y_1,\dots, Y_k)$ be the family of feasible solutions that attains $\opt{k}$. 
Let $\bm x$ be the vector in $\mathbb R^{\size{\mathcal X}}$ such that $x_{\alpha_i} = 1/k$ for each $1 \le i \le k$.
Otherwise, $x_j = 0$.
Since $\bm x$ is contained in $\Pi(\mathcal X)$,
the following inequality holds from the definition of $\opt{\infty}$.
\begin{align*}
    \opt{\infty}
    &\geq \sum_{1 \le i, j \le \size{\mathcal X}}x_ix_jd(X_i,X_j)
    = \frac{1}{k^2}\sum_{i=1}^{k}\sum_{j=1}^{k}d(Y_i,Y_j)\\
    &= \frac{2}{k^2}\sum_{1\leq i<j\leq k}d(Y_i, Y_{j})
    = \frac{k-1}{k}\opt{k}.
\end{align*}
We next give an upper bound of $\opt{\infty}$. 
Let $\mathcal D$ be a probability distribution that attains $\opt{\infty}$.
The following inequality and the lemma hold.
\begin{align*}
    \opt{\infty}
    &= \E_{(X_{1}, X_{2}) \sim \mathcal D \times \mathcal D}[d(X_{1}, X_{2})]
    = \E_{(X_{1}, \ldots, X_{k})\sim \mathcal D^k}\left[\frac{2}{k(k-1)}\sum_{1\leq i<j\leq k}d(X_{i},X_{j})\right]\\
    &\leq \max_{(X_{1},\dots, X_{k})\in \mathcal{X}^k}\left(\frac{2}{k(k-1)}\sum_{1\leq i<j\leq k}d(X_{i},X_{j})\right)
    = \opt{k}.
\end{align*}
Note that since $X_{i}$ and $X_{j}$ are sampled from $\mathcal D$ independently, $\E[d(X_{i}, X_{j})] = \E[d(X_{1}, X_{2})]$ for any $i$ and $j$, and the second equality holds.
\end{proof}

The above lemma shows that if we obtain a probability distribution $\mathcal D$ that attains $(1 - \varepsilon) \opt{\infty}$ in expectation, 
by performing $k$ random samples from it, a $(1 - \varepsilon)(1-1/k)$-approximate solution can be found.
In what follows, we show that we can obtain such a probability distribution by a local search approach.

We give an overview of our approach.
Let $\mathcal X = \set{X^{(1)}, \ldots, X^{(\size{\mathcal X})}}$ be the set of feasible solutions and
$P$ be the convex hull $\inset{\bm{1}^{(i)}}{X^{(i)} \in \mathcal{X}}$, where 
the vector $\bm 1^{(i)} \in \mathbb R^{\size{\mathcal X}}$ has $1$ $i$-th position and $0$ elsewhere.
Each point in the polytope $P$ can be identified with a probability distribution over $\mathcal X$.
For two points $x, x' \in P$, we define a function $\mu(x, x') \coloneqq \E_{(X, X') \sim x \times x'}[d(X, X')]$.
Notice that $\max_{x \in P}\mu(x, x) = \opt{\infty}$.
In what follows, we denote $\mu(x, x)$ as $\mu(x)$.
In our approach, we pick an arbitrary point $x^{(1)} \in \set{\bm{1}^{(1)}, \ldots, \bm{1}^{(\size{\mathcal X})}}$.
For each step, we find a vector $\bm 1^{(i)}$ that maximizes $\mu(x^{(1)}, \bm 1^{(i)})$, and 
update $x^{(1)}$ to $x^{(2)} = \mu(\lambda x^{(1)}, (1 - \lambda)\bm 1^{(i)})$, where $\lambda$ is a real number between $0$ and $1$
that maximizes $\mu(\lambda x^{(1)}, (1 - \lambda)\bm 1^{(i)})$.
By repeating this procedure sufficiently many times, we obtain a point in $P$ and
a desired solution by sampling $k$ solutions from this probability distribution.
See for the details in \Cref{alg:convex}.

\newcommand{\ceil}[1]{\lceil #1 \rceil}
\begin{algorithm}[t!]
\caption{A polynomial-time $(1 - \varepsilon)(1 - \frac{1}{k})$-approximation algorithm for \MSDC{} when duplication allowed. 
}\label{alg:convex}
\Procedure{\LocalDup{$\mathcal{X},\varepsilon$}}{
    Pick any $X^{(i)} \in \mathcal{X}$ and let $x\leftarrow \bm{1}^{(i)}$\;
    \For{$q=1,\dots, \ceil{16\varepsilon^{-1}}-1$}{
        Let $X^{(i)} \in \mathcal{X}$ be the argmax of $\mu(x,\bm{1}^{(i)})$\;\label{line:argmaxX}
        Let $\lambda$ be the argmax of $\mu(\lambda \bm{1}^{(i)}+(1-\lambda)x)$ between $0$ and $1$\;
        $x\leftarrow \lambda \bm{1}^{(i)}+(1-\lambda)x$\;
    }
    \Return $x$ as a convex combination of $\bm{1}^{(i)}$s for $X^{(i)} \in \mathcal{X}$\;
}
\end{algorithm}

To execute \Cref{alg:convex}, we have to determine the number of reputations of the above procedure,
and the following operations need to be performed efficiently.
\begin{enumerate}
    \item For a point $x\in P$, find a vector $\bm 1^{(i)}$ that maximizes $\mu(x, \bm 1^{(i)})$.
    \item For a point $x$ and a vector $\bm 1^{(i)}$, find $0 \le \lambda \le 1$ that maximizes $\mu(\lambda x, (1 - \lambda)\bm 1^{(i)})$.
\end{enumerate}

We show that if $d$ is a weighted Hamming distance, the operations~1 and 2 can be performed in ${\rm poly}(\size{E})$ time.
We first show the operation~1 can be done in ${\rm poly}(\size{E})$ time if we have a polynomial-time algorithm for finding a maximum-weight feasible solution in $\mathcal X$ for an arbitrary weight function $w'\colon E\to \mathbb R$.
To this end,
we transform $\mu(x, x')$ as follows.
\begin{align*}
    \mu(x, x') &\coloneq \E_{(X,X')\sim x \times x'}\left[d(X,X')\right]\\
    &= \E_{(X,X')\sim x\times x'}\left[w(X)+w(X')-2w(X\cap X')\right]\\
    &= \E_{X\sim x}\left[\sum_{e\in X}w(e)\right] + \E_{X'\sim x'}\left[\sum_{e\in X'}w(e)\right] - 2\E_{(X,X')\sim x\times x'}\left[\sum_{e\in X\cap X'}w(e)\right].
\end{align*}

Let $X$ be a solution in $\mathcal X$ drawn according to the distribution $x$, 
and let $p(x, e)$ denote the probability that element $e$ is included in $X$.
Then, $\E_{X \sim x} \left[ \sum_{e \in X} w(e) \right] = \sum_{e \in E} p(x, e) w(e)$ holds.
Therefore, we obtain the formula 
\begin{align*}
    \mu(x, x')
    &= \E_{X\sim x}\left[\sum_{e\in X}w(e)\right] + \E_{X'\sim x'}\left[\sum_{e\in X'}w(e)\right] - 2\E_{(X,X')\sim x\times x'}\left[\sum_{e\in X\cap X'}w(e)\right]\\
    &= \sum_{e\in E}w(e) (p(x, e) + p(x', e) - 2p(x, e)p(x', e))\\
    &= \sum_{e\in E}w(e) (p(x, e) + p(x', e) (1 - 2p(x, e))).
\end{align*}

When $x$ is fixed, to maximize $\mu(x, x')$, we want to maximize the factor $\sum_{e \in E}p(x', e)(1-2p(x, e))$.
If $x' = \bm{1}^{(i)}$ for some $i$, 
$\sum_{e \in E}p(x', e)(1-2p(x, e)) = \sum_{e \in X^{(i)}}(1 - 2p(x, e))$.
Therefore, if we have an algorithm that maximizes the sum of weight,
the argmax of line~\ref{line:argmaxX} can be found using this algorithm by setting the weight of each $e \in E$ to $w(e)(1 - 2p(x, e))$.
Notice that $w(e)(1 - 2p(x, e))$ is a constant since $x$ is fixed.

We next consider the operation~2, that is, to find $\lambda$ that maximizes $\mu(\lambda \bm 1^{(i)}, (1-\lambda)x)$.
To this end, we show the following auxiliary lemmas.

\begin{lemma}\label{lem:amgm}
For $x,x'\in P$, $\mu(x)+\mu(x') \leq 2\mu(x,x')$.
\end{lemma}
\begin{proof}
From the definition of $\mu(x, x')$, we obtain the following:
\begin{align*}
    \mu(x)+\mu(x')
    &= 2\sum_{e\in E}w(e)(p(x, e) - (p(x, e))^2 + p(x', e) - (p(x',e))^2)\\
    &\leq 2\sum_{e\in E}w(e)(p(x, e) + p(x', e) - 2p(x, e)p(x', e))
    = 2\mu(x,x').
\end{align*}
Notice that we use the inequality $a^2 + b^2 \ge 2ab$, where $a$ and $b$ are real numbers.
\end{proof}

\begin{lemma}\label{lem:dif_bound}
For $x \in P$, $\max_{X^{(i)} \in \mathcal{X}}\left(\mu(x,\bm{1}^{(i)})-\mu(x)\right) \geq (\opt{\infty}-\mu(x))/2$.

\end{lemma}
\begin{proof}
Let $x^*$ be the vector that attains $\opt{\infty}$. 
Then, $x^*$ can be written in the form of $x^*=\sum_{i=1}^{\size{\mathcal X}}\alpha^*_i \bm{1}^{(i)}$ where $\sum_{i=1}^{\size{\mathcal X}}\alpha^*_i = 1$ and
$0\leq \alpha^*_i\leq 1$.
Thus, we have
\begin{align*}
    \max_{i\in \set{1, \ldots, \size{\mathcal X}}}\left(\mu(x,\bm{1}^{(i)})-\mu(x)\right)
    &\geq \sum_{j=1}^{\size{\mathcal X}}\alpha^*_j \left(\mu(x,\bm{1}^{(j)})-\mu(x)\right)\\
    &= \mu(x,x^*) - \mu(x)\\
    &\geq \frac{\mu(x)+\mu(x^*)}{2} - \mu(x)\\
    &= \frac{\mu(x^*)-\mu(x)}{2}
    = \frac{\opt{\infty}-\mu(x)}{2}.
\end{align*}
The first inequality holds since
$\max_{i \in \set{1, \ldots, \size{\mathcal X}}} \mu(x, \bm{1}^{(i)}) \ge \sum_j \alpha^*_j \mu(x, \bm{1}^{(j)})$ holds.
The second inequality holds from \Cref{lem:amgm}.
The first equality holds from the bilinearity of $\mu$.
\end{proof}

Based on the two auxiliary lemmas above, we are now ready to prove the following lemma.

\begin{lemma}\label{lem:max:lambda}
    For fixed $x \in P$ and $\bm 1^{(i)}$, 
    the real number $\lambda$ that maximizes $\mu(\lambda \bm 1^{(i)}, (1-\lambda)x)$ is 
    \[
        \frac{\mu(\bm{1}^{(i)},x)-\mu(x)}{2\mu(\bm{1}^{(i)},x)-\mu(x)}.
    \]
\end{lemma}
\begin{proof}
    We consider a value of $\lambda$ that maximizes $f(\lambda) \coloneq \mu(\lambda \bm{1}^{(i)}+(1-\lambda) x)$.
    We obtain the following formula from the bilinearity of $\mu$:
    \begin{align*}
        \mu(\lambda \bm{1}^{(i)}+(1-\lambda) x)
        &= \lambda \mu(\lambda \bm{1}^{(i)}+(1-\lambda) x, \bm 1^{(i)}) + (1- \lambda) \mu(\lambda \bm{1}^{(i)}+(1-\lambda) x, x).
    \end{align*}
    
    By applying the same modification and $\mu(\bm{1}^{(i)}) = 0$ for any $i$, we obtain the followings.
    \begin{align*}
        \mu(\bm 1^{(i)}, \lambda \bm{1}^{(i)}+(1-\lambda) x)
        &= (1-\lambda)\mu(\bm 1^{(i)}, x)\\
        \mu(x, \lambda \bm{1}^{(i)}+(1-\lambda) x)
        &= \lambda \mu(x, \bm1^{(i)}) + (1 - \lambda)\mu(x)
    \end{align*}
    
    Thus, we obtain the followings:
    \begin{align*}
        f(\lambda)
        &= 2\lambda(1-\lambda) \mu(\bm{1}^{(i)},x) + (1-\lambda)^2 \mu(x)\\
        &= -\lambda^2 (2\mu(\bm{1}^{(i)},x)-\mu(x))
        + 2\lambda (\mu(\bm{1}^{(i)},x)-\mu(x))
        + \mu(x).
    \end{align*}

    We show that $f$ is a concave function.
    From Lemma~\ref{lem:dif_bound}, we have $\mu(\bm{1}^{(i)},x)-\mu(x)\geq 0$.
    From Lemma~\ref{lem:amgm}, we have $2\mu(\bm{1}^{(i)},x)-\mu(x) \geq \mu(\bm{1}^{(i)})+\mu(x)-\mu(x) = 0$.
    If $2\mu(\bm{1}^{(i)},x)-\mu(x) = 0$, that is, $\mu(x) = 2\mu(\bm{1}^{(i)},x)$,
    $\mu(\bm{1}^{(i)}, x) - \mu(x) = - \mu(\bm 1^{(i)}, x)$.
    Since $w(e)$ is non-negative for each $e \in E$, $\mu(\bm 1^{(i)}, x)$ is non-negative.
    Since $- \mu(\bm 1^{(i)}, x) \ge 0$, $\mu(\bm 1^{(i)}, x) = 0$.
    Therefore, $\mu(\bm 1^{(i)}, x) = \mu(x) = 0$.
    In this case, $f(\lambda) = 0$ for any $\lambda$, and we assume that $2\mu(\bm{1}^{(i)},x)-\mu(x) \neq 0$.
    This assumption is made without loss of generality.
    Thus, the quadratic coefficient of $f$ is negative.
    Therefore, $f$ attains its maximum value when $\lambda = (\mu(\bm{1}^{(i)},x)-\mu(x))/(2\mu(\bm{1}^{(i)},x)-\mu(x))$ since $f$ is a concave function.
\end{proof}

We next analyze the number of repetition and an approximation ratio of \Cref{alg:convex}.
To analyze these, we consider how much improvement is achieved in a single repetition.

\begin{lemma}
Let $x$ and $x'$ be points in $P$ obtained by the $q$-th loop and the $(q + 1)$-th loop of \Cref{alg:convex}.
Then, the following holds:
\begin{align*}
    \mu\left(x'\right)-\mu\left(x\right)
    \geq \frac{\left(\opt{\infty}-\mu\left(x\right)\right)^2}{16\opt{\infty}}.
\end{align*}
\end{lemma}
\begin{proof}    
    Let $X^{(i)}$ be a solution in $\mathcal X$ selected by the $q$-th loop.
    Using the fact that $\lambda = (\mu(\bm{1}^{(i)},x)-\mu(x))/(2\mu(\bm{1}^{(i)},x)-\mu(x))$ and 
    the bilinearity of $\mu$ as in \Cref{lem:max:lambda}, we obtain the following equation.    
    \begin{align*}
        \mu(x') - \mu(x)
        &= -\frac{(\mu(\bm{1}^{(i)},x)-\mu(x))^2}{2\mu(\bm{1}^{(i)},x)-\mu(x)} + \frac{2(\mu(\bm{1}^{(i)},x)-\mu(x))^2}{2\mu(\bm{1}^{(i)},x)-\mu(x)}\\
        &= \frac{(\mu(\bm{1}^{(i)},x)-\mu(x))^2}{2\mu(\bm{1}^{(i)},x)-\mu(x)}\\
        &\ge \frac{(\opt{\infty} - \mu(x))^2}{4(2\mu(\bm{1}^{(i)},x)-\mu(x))}.
    \end{align*}
    
    The last inequality holds from \Cref{lem:dif_bound}.
    Moreover, 
    the following inequality shows that $2\mu(\bm{1}^{(i)},x) - \mu(x)$ is bounded by $4\opt{\infty}$.
    \begin{align*}
        2\mu(\bm{1}^{(i)},x)-\mu(x)
        \leq 2\mu(\bm{1}^{(i)},x)
        \leq 2\max_{X^{(j)}\in \mathcal{X}}\mu(\bm{1}^{(i)},\bm{1}^{(j)})
        \leq 2\opt{2}
        \leq 4\opt{\infty},
    \end{align*}
    where the second inequality is from the facts that $\mu$ is bilinear and $x$ can be expressed as a convex combination of elements of $\mathcal{X}$ and the last inequality is from Lemma~\ref{lem:kandinf}.
    Therefore,
    the statement holds.
\end{proof}

\begin{lemma}
Let $x$ be a point in $P$ obtained by the $q$-th loop of \Cref{alg:convex}.
Then, 
\begin{align*}
    \mu\left(x\right) \geq \left(1-\frac{16}{q}\right)\opt{\infty}.
\end{align*}
\end{lemma}
\begin{proof}
We prove this by induction. The lemma is trivial for $q=1$. 
Let $x'$ be a point obtained by $q + 1$-th loop.
We have
\begin{align*}
    \opt{\infty} - \mu\left(x'\right)
    &\leq \opt{\infty}-\left(\mu\left(x\right)+\frac{\left(\opt{\infty}-\mu\left(x\right)\right)^2}{16\opt{\infty}}\right)\\
    &= \left(\opt{\infty}-\mu\left(x\right)\right)\left(1-\frac{\opt{\infty}-\mu\left(x\right)}{16\opt{\infty}}\right)\\
    &\leq \frac{16}{q}\left(1-\frac{1}{q}\right) \opt{\infty}\\
    &\leq \frac{16}{q+1}\opt{\infty},
\end{align*}
where the last inequality follows that $z\left(1-\frac{z}{16}\right)$ is increasing for $0 \le z \le 1$.
Therefore, the statement holds.
\end{proof}

From the above lemma, by repeating local improvement $\lceil 16 \varepsilon^{-1} \rceil$ times,
we obtain a desired point in $P$.
We analyze the complexity of \Cref{alg:convex}.
Suppose that we have an algorithm for finding a maximum-weight solution $\mathcal X$ for an arbitrary weight function $w'\colon E \to \mathbb R$ that runs in $\order{\size{E}^c}$ time.
Picking an arbitrary solution $X^{(1)}$ can be done using this algorithm, 
and we use this algorithm once in each repetition. 
The other operations in the loop can be done in $\order{\varepsilon^{-1}}$ time.
Therefore, the running time of this algorithm is $\order{\varepsilon^{-2}\size{E}^c + k}$ time.

\begin{theorem}\label{thm:expectation}
    Let $E$ be a finite set, $\mathcal X \subseteq 2^E$, and $w\colon E \to \mathbb R_{> 0}$.
    Suppose that there is an algorithm for finding a solution $X \in \mathcal X$ with the maximum weight for an arbitrary weight function $w' \colon E \to \mathbb{R}$ that runs in $\order{\size{E}^c}$ time.
    Then, there is a randomized algorithm for \MSDC{} with an approximation factor $(1 - \varepsilon)(1 - 1/k)$ for any $\varepsilon > 0$ in expectation that runs in $\order{\varepsilon^{-2}\size{E}^c + k}$ if we allow duplication.
\end{theorem}

Finally, we derandomize \Cref{alg:convex}.
Let $x$ be a probability distribution founded by \Cref{alg:convex} and $\mathcal X(x)$ be the set of elements in $\mathcal X$ with non-zero probability.
To derandomize this algorithm, the only change is to deterministically choose $k$ elements from the non-zero entries of $x$ that maximize $d_{\rm sum}$, rather than sampling them at random from the distribution $x$.

\Cref{alg:convex} updates a probability distribution by computing a convex combination of the current probability distribution and
indicator vector that maximizes the expected value at each step.
Since the number of updates is $\order{\varepsilon^{-1}}$, 
the number of non-zero elements in $x$ is bounded by $\order{\varepsilon^{-1}}$.
Therefore, $k$ elements in $\mathcal X(x)^k$ that maximizes the function $d_{\rm sum}$ can be found in $\order{\varepsilon^{-(k+2)}\size{E}}$ time.
Notice that the additional factor $2$ is the time needed to evaluate the function value of $d_{\rm sum}$.
When selecting $k$ elements in a depth-first manner, each update of $d_{\rm sum}$ can be done in $\order{k\size{E}}$ time, and
a desired set of solutions can be found in $\order{\varepsilon^{-(k+1)}\size{E}}$ time.
Therefore, we obtain the following theorem.

\begin{theorem}
    Let $E$ be a finite set, $\mathcal X \subseteq 2^E$, and $w\colon E \to \mathbb R_{> 0}$.
    Suppose that there is an algorithm for finding a solution $X \in \mathcal X$ with the maximum weight for an arbitrary weight function $w' \colon E \to \mathbb{R}$ that runs in $\order{\size{E}^c}$ time.
    Then, there is a deterministic algorithm for \MSDC{} with an approximation factor $(1 - \varepsilon)(1 - 1/k)$ for any $\varepsilon > 0$ that runs in $\order{\varepsilon^{-(k+1)}\size{E}^{c+1}}$ time if we allow duplication.    
\end{theorem}

\section{Conclusion}\label{sec:concl}
We give a framework for designing approximation algorithms for \MSDC{}. 
This framework runs in $\mathrm{poly}(\size{E} + k)$ time and is versatile, which allows to apply to the diverse version of several well-studied combinatorial problems.
The key to applying our framework is a polynomial-time algorithm for \textsc{Weighted Extension}, which yields constant-factor approximation algorithms for \MSDM{} and \DMCB{}.
Moreover, we obtain a PTAS for \MSDC{} if we can solve the problem in polynomial time for fixed $k$, yielding PTASes for \DMGC{} and \DIS{}.

There are several directions from our work.
Our approximation algorithms for \MSDM{} and \DMCB{} give a approximation factor $\max(1-2/k, 1/2)$, which is a constant when $k$ is a constant.
The APX-hardness of these problems is an interesting question to prove a limitation of finding ``approximately'' diverse solutions.
Our work focuses only on \textsc{Max-Sum Hamming Distance} as our objective function.
However, \textsc{Max-Min Hamming Distance} or other diversity measures would be more acceptable in some practical applications.
It would be worth investigating these diversity measures from the viewpoint of approximability.

\subsubsection*{Acknowledgements}
The authors thank Yutaro Yamaguchi, Jaehoon Yu, and Shungo Kumazawa for helpful discussions which led to improvements of the paper. 
This work was partially supported by JSPS Kakenhi Grant Numbers 
JP19H01133, 
JP21K17812 and 
JP21H05861, 
JST CREST Grant Number JPMJCR18K3, 
and JST ACT-X Grant Number JPMJAX2105, 
Japan

\bibliographystyle{plain}
\bibliography{main}

\end{document}